\documentclass[envcountsame,a4paper]{llncs}  
\pdfoutput=1

\usepackage{amssymb}

\usepackage[UKenglish]{babel}
\usepackage{microtype}
\usepackage{lmodern}

\usepackage{mathtools}

\usepackage{standalone}

\usepackage{xspace}

\usepackage{graphicx} 

\usepackage{etoolbox} 

\usepackage{braket}  

\usepackage{tikz}
\usetikzlibrary{arrows,positioning,matrix,calc,trees}
\usetikzlibrary{shapes.geometric,backgrounds}

\usepackage{verbatim}

\usepackage[textwidth=1.8cm,textsize=footnotesize,color=blue!40]{todonotes}

\usepackage{hyperref}

\usepackage{letltxmacro}  
\LetLtxMacro{\Standardrestriction}{\restriction}     
\renewcommand*{\restriction}[1]{\Standardrestriction_{#1}}

\LetLtxMacro{\Standardexists}{\exists}    
\renewcommand*{\exists}[1]{\Standardexists #1 \,}
\LetLtxMacro{\Standardforall}{\forall}
\renewcommand*{\forall}[1]{\Standardforall #1 \,}

\let\Latextextbf\textbf
\renewcommand{\textbf}[1]{\Latextextbf{\boldmath #1}}

\newcommand{\define}[1]{\emph{#1}}

\DeclarePairedDelimiter{\abs}{\lvert}{\rvert}

\newcommand{\Buchi}{B{\"u}chi\xspace}
\newcommand{\Bojanczyk}{Boja\'nczyk\xspace}
\newcommand{\Torunczyk}{Toru\'nczyk\xspace}

\newcommand*{\TITEL}{Model Checking Constraint LTL over Trees}
\newcommand*{\KEYWORDS}{\Buchi automata, constraint automata, concrete domains,
  constraint LTL}
\newcommand{\AUTHORS}{%
Alexander Kartzow \texorpdfstring{\inst{1,2} \and}{ and }Thomas Weidner\texorpdfstring{\inst{1}}{}
}

\DeclareMathOperator{\bigsqcap}{ 
  \rotatebox[origin=c]{180}{\ensuremath{\bigsqcup}}} 

\newcommand{\N}{\mathbb{N}}
\newcommand{\Z}{\mathbb{Z}}
\newcommand{\Q}{\mathbb{Q}}

\newcommand{\Struc}[1]{\mathfrak{#1}}   
\newcommand{\Aut}[1]{\mathbb{#1}}  

\DeclareMathOperator{\image}{im} 
\DeclareMathOperator{\dom}{dom}

\newcommand{\neX}{\mathsf{X}}
\newcommand{\Glob}{\mathsf{G}}

\newcommand{\Until}{\mathsf{U}}

\newcommand{\LTL}{\ensuremath{\mathrm{LTL}}\xspace}
\newcommand{\CTLStar}{\mathrm{CTL}^*}
\newcommand{\cLTL}{\ensuremath{\mathrm{LTL}(\set{\preceq, \KBleq, S})\xspace}}
\newcommand{\Var}{\mathcal{V}}
\newcommand{\PSPACE}{\mathrm{PSPACE}}
\newcommand{\NPSPACE}{\mathrm{NPSPACE}}

\newcommand{\stretchleq}{\mathrel{\leq_C}}

\DeclareMathOperator{\typ}{typ}

\DeclareMathOperator{\MCAT}{MCAT}

\newcommand*{\RunTypes}{{\ensuremath{\mathsf{Typs}_{n,C}}}\xspace}
\newcommand*{\Tree}[1][\infty]{\Struc{T}_{#1}}
\newcommand*{\KBleq}{\mathrel{\sqsubseteq}}
\newcommand*{\KBless}{\mathrel{\sqsubset}}

\newcommand*{\SAT}[1]{\operatorname{SAT}(#1)}
\newcommand*{\MC}[1]{\operatorname{MC}(#1)}

\hypersetup{%
pdfauthor={\AUTHORS},
pdftitle={\TITEL},
pdfkeywords={\KEYWORDS}
}

\begin{document}
\title{\TITEL\thanks{This work is supported by the DFG Research Training
    Group 1763 (QuantLA) and the DFG
    research project GELO.}}

\institute{Institut f\"ur Informatik, Universit\"at Leipzig, Germany \and
  Department f\"ur Elektrotechnik und Informatik, Universit\"at
  Siegen, Germany}

\author{\AUTHORS}
\maketitle

\begin{abstract}
  Constraint automata are an adaptation of \Buchi-automata that process
data words where the data comes from some relational structure
$\Struc{S}$. Every transition of such an automaton comes with
constraints in terms of the relations of $\Struc{S}$. A transition 
can only
be fired if the current and the next data values satisfy all
constraints of this transition. These automata have been used in the
setting where $\Struc{S}$ is a linear order for 
deciding constraint $\LTL$ with constraints over $\Struc{S}$.
In this paper, $\Struc{S}$ is the  infinitely
branching infinite order tree $\Struc{T}$. We 
provide a $\PSPACE$ algorithm for emptiness of $\Struc{T}$-constraint
automata. This result implies
$\PSPACE$-completeness of the satisfiability and the model checking
problem for  constraint $\LTL$ with constraints over $\Struc{T}$.


\end{abstract}

\section{Introduction}
Temporal logics like $\LTL$  or
$\CTLStar$ are nowadays 
standard languages for specifying system properties in
verification. These logics are interpreted over node labelled graphs,
where the node labels (also called atomic 
propositions) represent abstract properties of a system (for
instance,  a computer program). Clearly,
such an abstracted system state does not in  general contain all the
information of the original system state. This may lead to incorrect
results in model checking. 

In order to overcome this weakness, extensions of temporal logics by
atomic (local) constraints over some structure $\Struc{A}$  have been
proposed (cf.~\cite{Cerans94,DemriG08}). 
For instance, \emph{\LTL with local constraints}
is evaluated over infinite words where the letters are tuples over
$\Struc{A}$ of a fixed size.  For instance, for 
$\Struc{A}= (\Z, <)$, this logic is  standard $\LTL$ where atomic
propositions are replaced by atomic constraints of the
form $\neX^i x_j < \neX^l x_k$. This
constraint is satisfied by  
a path $\pi$ if 
the $j$-th element of the $i$-th letter 
of $\pi$ is less than the $k$-th element of the $l$-th letter of $\pi$.

While temporal logics with integer constraints are suitable to
reason about programs manipulating counters,
reasoning about systems manipulating pushdowns requires  constraints
over words over a fixed alphabet and the 
prefix relation (which is equivalent to constraints over an infinite
$k$-ary tree with descendant/ancestor relations). 
There are numerous investigations on satisfiability and model
checking for temporal logics with constraints over the integers
(cf.~\cite{Cerans94,BozzelliG06,DemriG08,Gascon09,BozzelliP14,CKL14}). Contrary, 
temporal logics with constraints over trees have not yet been
investigated much, although questions concerning decidability of the
satisfiability problem for $\LTL$ or $\CTLStar$ with
such constraints have been asked for instance in
\cite{DemriG08,CKL13}. A first (negative) result by Carapelle et
al.~\cite{CFKL14} shows that a technique developed in 
\cite{CKL13,CKL14} for 
satisfiability results of
branching-time logics (like $\CTLStar$ or $\mathrm{ECTL}^*$) with
integer constraints cannot be used 
to resolve the satisfiability status of temporal logics with
constraints over trees. 

Our goal is to show that satisfiability of $\LTL$ with
constraints over the tree is decidable. At first,
we analyse the emptiness problem of
$\Struc{T}$-constraint automata (cf.~\cite{Gascon09,DemriD07}) where
$\Struc{T}$ is the infinitely branching infinite tree with prefix
relation.  These 
automata are \Buchi-automata that process (multi-)data words where the
data values are elements of $\Struc{T}$
where applicability of transitions  
depends on the order of the data values at the current and the next 
position.
Our technical main
result shows that emptiness for these automata is $\PSPACE$-complete.
Having obtained an algorithm for the emptiness problem, we can easily 
provide algorithms for the 
satisfiability and model checking problems for 
$\LTL$ with constraints over $\Struc{T}$. We exactly mimic the 
automata based algorithms for standard $\LTL$ of Vardi and Wolper
\cite{VardiW94} noting that the constraints in the transitions are
exactly what is needed to deal with the atomic constraints in the
local constraint  version of
$\LTL$. It follows directly that  satisfiability of 
$\LTL$ with constraints over $\Struc{T}$ and 
model checking models defined by constraint automata against
$\LTL$ with constraints over $\Struc{T}$ is $\PSPACE$-complete. 

Finally, we extend our results to the case of constraints over the
infinite $k$-ary tree for every $k\in\N$ by providing
a reduction  to
$\LTL$ with constraints over $\Struc{T}$. 
Thus, satisfiability and model checking for  $\LTL$ with
constraints over the infinite $k$-ary tree is also in $\PSPACE$.

Upon finishing our paper, we have become aware that Demri and Deters
(abbreviated DD in the following)
have submitted a paper \cite{DemriD14} that shows above mentioned
results on satisfiability using a reduction of constraints over trees
to constraints over the integers.
Even though the
main results of both papers coincide, there are major differences.
\begin{enumerate}
\item  DD's result extends to satisfiability of the
  corresponding version of $\CTLStar$, but DD do not consider
  the model checking problem.
\item DD's result holds
  even if the logic is enriched by length constraints that compare the 
  lengths of the interpretations of variables. Since our approach
  abstracts away the concrete length of words, we cannot reprove this
  result. On the other hand, 
  we can
  enrich the logic  with  constraints using the
  lexicographic order on the tree as well. DD's approach can not deal
  with this order.  Thus, the logic in each paper is incomparable to
  the logic of the other.
\item DD conjecture that (branching-degree) uniform satisfiability
  problem is in $\PSPACE$. This problem asks, given a formula and a
  $k\in\N\cup\set{\infty}$ whether there is a model with values in the
  $k$-ary infinite tree that satisfies the formula. 
  We  confirm DD's conjecture.  
\item Finally, our proof is self-contained.
  In contrast, DD's proof seems to be more elegant
  and less technical, but this comes at the cost of relying on the
  decidability result for satisfiability of $\LTL$ with constraints
  over the integers \cite{BozzelliP14}, which is again quite technical to
  prove.\footnote{In fact, our proof can be easily adapted to reprove
    this result.}
\end{enumerate}

Our result leaves open several further research directions. 
Firstly, DD's result on $\CTLStar$ with constraints over trees does
not yield any reasonable complexity bound because the complexity of
their algorithm
relies on the results of \Bojanczyk and
\Torunczyk \cite{BojanczykT12} on weak monadic second order logic with
the unbounding quantifier. Thus, without any progresses concerning the
complexity of this logic, DD's approach cannot be used to obtain
better bounds. In contrast, the concept of $\Struc{T}$-constraint
automata can be 
easily lifted to a $\Struc{T}$-constraint tree-automaton
model. Complexity bounds on the 
emptiness problem for this model would directly imply bounds on the
satisfiability for $\CTLStar$ with constraints over $\Struc{T}$. 
Thus, investigating whether our approach transfers to a result on the
emptiness problem of $\Struc{T}$-constraint tree-automata might be a
fruitful approach. 
Secondly, it may be possible to lift our results to the global model
checking problem similar to  the work of Bozelli and Pinchinat
\cite{BozzelliP14} on $\LTL$ with constraints over the integers.
Finally, it is a very challenging task to decide whether DD's result
and our result can be unified to a result on $\LTL$ with constraints
over the tree with prefix order, lexicographic order and
length-comparisons (of maximal common prefixes).



\section{Model Checking LTL with Constraints over Trees}

We first introduce $\cLTL$, a variant of $\LTL$ with local
constraints. A model of a formula of this logic is a (multi-) data
word where the data comes from some $\set{\preceq,\KBleq,
  S}$-structure. We are particularly interested in the case where this
structure is an order tree with lexicographic order $\KBleq$. We want
to adjust the automata-based model checking methods for $\LTL$ to this
setting. For this purpose we then recall the definition of
tree-constraint automata.  The technical core of this paper shows that
emptiness of tree-constraint automata is $\PSPACE$-complete. Before we
delve into this technical part, we prove that satisfiability and
model checking for $\cLTL$ formulas with constraints over the full
infinitely branching tree are in $\PSPACE$ due to a reduction to the
emptiness problem of tree-constraint automata. We conclude this
section by providing a reduction of satisfiability and model checking
for $\cLTL$ with constraints over the full tree of branching degree $k$
to the corresponding problem over the full infinitely branching tree.

\subsection{LTL with Constraints}

Constraint $\LTL$ over signature $\set{{=}, {\preceq}, {\KBleq}, s_1, s_2,\dots, s_m}$
where $S = \Set{s_1, \dots, s_m}$ is a set of constant symbols,
abbreviated $\cLTL$, is given by the grammar
\begin{equation*}
  \phi  ::=  \neX^i x_1 \mathrel{*} s \mid  s \mathrel{*} \neX^i x_1 \mid
  \neX^{i} x_1 \mathrel{*} \neX^{j} x_2 \mid
  \neg \phi \mid (\phi \land \phi) \mid \neX\phi \mid \phi \Until \psi
 \mid \Glob \phi
\end{equation*}
where $* \in \Set{=, \preceq, \KBleq}$, $i,j$ are natural
numbers, $x_1, x_2$ are 
variables from some countable fixed set $\Var$ and
$s \in S$ is a constant symbol.  
Given a structure $\Struc{A} = (A,
\preceq^{\Struc{A}}, \KBleq^{\Struc{A}} s_1^{\Struc{A}}, s_2^{\Struc{A}}, \dots,
s_m^{\Struc{A}})$, 
an \define{$n$-dimensional data word over
  $\Struc{A}$} is a sequence
$(\bar a_i)_{i\in \N}$ with $\bar a_i\in A^n$.
We evaluate 
a formula $\phi$ (where $x_1, \dots,
x_n\in \Var$ are the variables occurring in $\phi$)
on $n$-dimensional data words $(\bar a_i)_{i\in \N}$.
We write $a_i^j$ for the $j$-th
component of $\bar a_i$.
We say \define{$(\bar a_i)_{i\in\N}$ is a model of $\phi$},
denoted as $(\bar a_i)_{i\in\N} \models \phi$, 
if the  usual conditions for $\LTL$ hold, and 
the following additional rules apply for
$*\in \Set{=, \preceq, \KBleq}$:
\begin{itemize}
\item  $(\bar a_i)_{i\in\N} \models (\neX^{i} x_k) \mathrel{*}
  (\neX^{j} x_l)$ if and only if $\Struc{A} \models a_{i}^l \mathrel{*} a_{j}^k$,
\item $(\bar a_i)_{i\in\N} \models (\neX^i x_l) \mathrel{*} s_j$ 
  (or $s_j \mathrel{*} (\neX^i x_l)$, resp.) 
  if and only if
  $\Struc{A} \models a_{i}^l \mathrel{*} s_j$ (or
  $\Struc{A} \models s_j \mathrel{*} a_{i}^l$, respectively).
\end{itemize}
Note that our constraint LTL does not use
atomic propositions. On nontrivial structures, 
proposition $p$ can be resembled
by constraints of the form  $x_{p_1} = x_{p_2}$.

As for usual $\LTL$ one defines dual operators. Then every formula
has an equivalent negation normal form where negation only appears in
front of atomic constraints 
($(\neX^{i} x_1) \preceq (\neX^{j} x_2)$,
$s \preceq \neX^i x$
or $\neX^i x\preceq s $). 
Using that $X^n(X^i x_k \ast X^j x_{\ell}) \equiv X^{i+n} x_k \ast X^{j+n} x_{\ell}$ and by introducing auxiliary variables,
it is also easy to  eliminate 
exponents in terms:

\begin{proposition} \label{prop:Depth2ConstraitnsSuffice}
  There is a polynomial time algorithm that computes, on input a
  $\cLTL$-formula $\phi$ an equivalent $\cLTL$-formula $\psi$ such
  that $\psi$ does not contain terms of the form $\neX^i x$ with
  $i\geq 2$. 
\end{proposition}

We want to investigate $\cLTL$ in the cases where the structure
$\Struc{A}$  is one of the following order trees.
For each $k\in \set{2,3,4,\dots }$, let
\begin{equation*}
  \Tree^C = (\Q^*, \preceq, \KBleq c_1, c_2,\dots, c_m)\text{ and }
  \Tree[k]^C = (\{1, 2,\dots, k\}^*, \preceq,\KBleq,  c_1, \dots, c_m)
\end{equation*}
where $\preceq$ is the prefix order, $\KBleq$ is the lexicographic
order defined by $w \KBleq v$ if either $w \preceq v$ or there are
$q_1,q_2\in Q$ such that $(w\sqcap v)q_1 \preceq w$, $(w\sqcap v)q_2
\preceq v$ and $q_1 < q_2$, where $<$  is the natural order on $\Q$ and $\sqcap$ denotes the (binary) \define{greatest common prefix
operator}, and
$C=(c_1,c_2, \dots c_m)$ is a tuple of constants in $\Q^*$ or
$\set{1,2,\dots, k}^*$, respectively.

\subsection{Constraint Automata}
\label{sec:constAut}

In the following, we investigate the satisfiability and model checking
problems for $\cLTL$ over models with data values in one of the trees
$\Tree[k]^C$ for $k\in\set{\infty, 2,3,4,\dots}$. We follow closely
the automata theoretic approach of Vardi and Wolper \cite{VardiW94}
which provides a reduction of model checking for $\LTL$ to the
emptiness problem of \Buchi automata. In order to deal with the
constraints, we use \emph{$\Tree[k]^C$-constraint automata}
(cf.~\cite{Gascon09}) instead of \Buchi automata. Next we recall the
definition of constraint automata and state our main result concerning
emptiness of constraint automata. We then derive analogous results of
Vardi and Wolper's decidability results on \LTL for \cLTL with
constraints over $\Tree[k]^C$. A
\emph{$\Tree[k]^C$-constraint automaton} is defined as a usual
\Buchi automaton but instead of labelling transitions by some letter
from a finite alphabet we label them by Boolean combinations of
constraints which the current and the next data values have to satisfy
in order to apply the transition. 

\begin{definition}
  \begin{itemize}
  \item An \define{$n$-dimensional
      $\Tree[k]^C$-constraint automaton} 
    is a   quadruple 
    \mbox{$\Aut{A} = (Q, I, F, \delta)$}
    where $Q$ is a finite set of states, $I\subseteq Q$ the initial
    states, $F\subseteq Q$ the set of accepting states and
    $\delta \subseteq Q \times B^C_n \times Q$ the \define{transition
      relation} where $B^C_n$ is the set of all quantifier-free
    formulas over signature $\set{\preceq, \KBleq}\cup C$ with
    variables $x_1, \dots, x_n, y_1, \dots, y_n$, i.e., propositional
    logic formulas with atomic formulas $v \ast v'$, with
    ${\ast} \in \{{=}, {\preceq}, {\KBleq}\}$ and $v$, $v'$ are
    variables or constants.
  \item A
    \define{configuration} of the automaton $\Aut{A}$ is a tuple
    in $Q \times (\set{1,2,\dots,k}^*)^n)$  (or $(\Q^*)^n$ if $k = \infty$).
  \item
    We define $(q, \bar w) \to (p, \bar v)$ 
    iff there is a transition $(q, \beta(x_1, \dots x_n, y_1, \dots,
    y_n), p)$ such that $\Tree[k]^C \models \beta(\bar w, \bar v)$.
  \item A \define{run} of $\Aut{A}$ is a finite or infinite sequence of
    configurations $r = (c_j)_{j\in J}$ ($J\subseteq \N$ an interval)
    such that $c_j \to c_{j+1}$ for all $j,j+1 \in J$.  For a finite
    run $r = (c_i)_{ i_1 \leq i \leq i_2}$ with $i_1 \leq i_2 \in \N$
    we say $r$ is a \emph{run from $c_{i_1}$ to $c_{i_2}$}.
  \item A   run $r=(c_i)_{i\in\N}$ is \emph{accepting} if 
    $c_0 = (q,d_1, \dots, d_n)$ for an initial state $q\in I$ and a final
    state $f\in F$ appears in infinitely many configurations
    of $r$.
  \item The \define{set 
    of all words accepted by $\Aut{A}$}
    comprises all $\bar{w}_1 \bar{w}_2 \dots \in  ((\Q^*)^n)^{\omega}$ (or $(\{1,\dotsc,k\})^n)^{\omega}$ if $k \ne \infty$)
    such that there  
    is an accepting infinite run $(c_i)_{i \in \N}$ with 
    $c_i = (q_i, \bar{w}_i)$.
  \end{itemize}
\end{definition}

In the following sections (see Theorem \ref{thm:EmptinessPSPACE}) we prove
that emptiness of $n$-dimensional 
$\Tree^C$-constraint automata is $\PSPACE$-complete in terms of
$\abs{Q} + \abs{C} + k + m$ where $m$ is the length of the longest
constant occurring in $C$. We next 
apply this result in order to obtain
$\PSPACE$-completeness of  satisfiability and model checking.

\subsection{Satisfiability and Model Checking of Constraint LTL}

\begin{definition}
  Let $k\in \set{\infty,2,3,4,\dots}$. 
  
  $\SAT{\Tree[k]^C}$ denotes the \define{satisfiability problem} for
  $\cLTL$ over 
  $\Tree[k]^C$: given a set of constants $C$ and a $\cLTL$-formula $\varphi$, is
  there a data word $(\bar w_i)_{i\in\N}$ over   $\Tree[k]^C$ such
  that $(\bar w_i)_{i\in\N} \models \varphi$?

  $\MC{\Tree[k]^C}$ denotes the \define{model checking problem for 
  $\Tree[k]^C$-constraint automata against $\cLTL$}:
  given a set of constants $C$, a $\Tree[k]^C$-constraint automaton $\Aut{A}$ and a
  $\cLTL$-formula $\varphi$, is 
  there a data word $(\bar w_i)_{i\in\N}$ over $\Tree[k]^C$ accepted
  by $\Aut{A}$ such that $(\bar w_i)_{i\in\N} \models \varphi$?
\end{definition}

\begin{theorem} \label{thm:SatAndMCInPSPACE}
  Let $k\in \set{\infty,2,3,4,\dots}$ and $C$ a set of constants. 
  $\SAT{\Tree[k]^C}$ and $\MC{\Tree[k]^C}$ are $\PSPACE$-complete.
\end{theorem}
\begin{proof}
  Since there is an automaton accepting all data words, the
  satisfiability problem reduces to the model checking problem whence
  it suffices to prove the claim on model checking. Hardness follows
  directly from the known results for $\LTL$. 
  We first prove $\MC{\Tree^C}\in\PSPACE$ and then we provide a
  reduction of $\MC{\Tree[k]^C}$ to  $\MC{\Tree^C}$ for all other
  $k$. 

  \smallskip
  \noindent \textbf{Case $k = \infty$.}
  Let $C\subseteq \Q^*$ be a finite set of constants, $\Aut{A}$ a 
  $\Tree^C$-constraint automaton and $\varphi\in \cLTL$.
  Due to Proposition \ref{prop:Depth2ConstraitnsSuffice} we can assume
  that all atomic constraints occurring in $\varphi$ only concern the current
  and the next data values. 
  Recall that Vardi and Wolper~\cite{VardiW94} provided a translation
  from $\LTL$ to \Buchi automata such that the resulting automaton
  accepts some word if and only if it is a model of the formula.
  
  This
  translation directly lifts to a translation of $\cLTL$ over $\Tree$ 
  to  $\Tree$-constraint automata.  As in the
  standard construction, each state of the automaton is
  a subset of (the negation closure of) the set of subformulas of the
  $\cLTL$-formula. Intuitively, an accepting run of the automaton on
  $(\bar w_i)_{i\in\N}$ is at position
  $i_0$ in a state containing some subformula $\psi$ if and only if 
  $(\bar w_i)_{i\geq i_0} \models \psi$. Obviously the dependence of
  the transitions of a 
  constraint automaton on the order of the current and next data
  values is exactly what is needed to allow the automaton to switch from
  one state to another only if the (possibly negated) atomic constraints
  contained in the current state are satisfied by the current and the
  next data values. 

  Thus, we obtain a constraint automaton $\Aut{B}$ such that $\Aut{B}$
  accepts $(\bar w_i)_{i\in\N}$ if and only if
  $(\bar w_i)_{i\in\N}\models \varphi$. Since the usual product
  construction for \Buchi automata lifts also to constraint automata,
  we easily construct in polynomial space an automaton $\Aut{C}$ such
  that $\Aut{C}$ accepts a word if and only if both $\Aut{A}$ and
  $\Aut{B}$ accept this word. Thus, the set of all words accepted by
  $\Aut{C}$ is non-empty if and only if there is a data word
  $(\bar w_i)_{i\in\N}$ such that $\Aut{A}$ accepts
  $(\bar w_i)_{i\in\N}$ and $(\bar w_i)_{i\in\N}\models\varphi$.
  Since emptiness is in $\PSPACE$ the claim follows.
	
  \smallskip
  \noindent \textbf{Case $k \neq \infty$.}
  Now we turn to the case $\Tree[k]^C$ where $k\neq \infty$. 
  Let  $C_l$ be the set of
  $\preceq$-maximal elements of $C$, and let $\varphi$ and $\Aut{A}$
  as before. 
  Without loss of generality we can assume that $C_l$ intersects every
  infinite branch in $\set{1,2,\dots, k}^\omega$
  (%
  If not, add $ci$ as
  a new constant for 
  every $c$ in the prefix-closure of $C$ and $i\in\set{1,2,\dots, k}$, which only causes a
  polynomial growth of the input%
  ).
  We claim that $(C, \Aut{A}, \varphi)$ is a positive instance of
  $\MC{\Tree[k]^C}$ if and only if 
  $(C, \Aut{A}, \psi)$ is a positive instance of
  $\MC{\Tree^C}$ where $\Aut{A}$ is seen as a $\Tree^C$-automaton and 
  $\psi  =\varphi \land
  \Glob \bigwedge_{i=1}^n \bigvee_{c\in C_l} (x_i \preceq c \lor
  c\preceq  x_i)$ where $x_1,x_2, \dots, x_n$ is the set of variables
  occurring in the constraints of $\varphi$. Basically, $\psi$ is
  $\varphi$ with the additional condition that the data values
  occurring in a model form a tree of branching degree $k$ at all
  constants. 
  It is clear that every witness $(\bar w_i)_{i\in\N}$ for the former
  model checking problem is a witness for the latter. 

  For the converse
  assume that $(\bar w_i)_{i\in\N}$ is a data word over $\Tree$
  accepted by $\Aut{A}$ satisfying $\psi$. 
  Note that there is an injective map $g:\Q^*\to \set{1,2}^*$ preserving
  $\preceq$ and $\KBleq$ in both directions (cf. Appendix
  \ref{sec:Orders}). Moreover, by definition of $\psi$ we 
  conclude that every value occurring in $(\bar w_i)_{i\in\N}$ is
  either a prefix of one of the constants or of the form
  $c q_1q_2\dots q_n$ for some maximal constant $c\in C_l$. 
  Thus, we can define $\bar v_i = (v_i^1, v_i^2, \dots, v_i^n)$ where 
  $v_i^j= w_i^j$ if $w_i^j \preceq c$ for some $c\in C_l$ and
  $v_i^j = c g(u)$ if 
  $w_i^j = c u$  for some $c\in C_l$ and $u\neq \varepsilon$. 
  Clearly $(\bar v_i)_{i\in\N}$ is a data word over $\Tree[k]$. 
  Since $g$ preserves $\preceq$, $\KBleq$ and all constants, it is a
  model of $\psi$ accepted by $\Aut{A}$ whence it is also a model of
  $\varphi$. 
\qed
\end{proof}
\begin{remark}
  Demri and Deter \cite{DemriD14} conjectured that if the arity $k$ of
  the tree is part of the input to the satisfiability problem, it is
  still in $\PSPACE$.  Our proof confirms that this
  branching degree uniform satisfiability problem is
  $\PSPACE$-complete.
\end{remark}

\section{Emptiness of Tree Constraint Automata}

Recall that every nonempty \Buchi automaton has an accepting
run which is ultimately periodic. We first prove that a nonempty  
constraint automaton has an accepting run which ultimately consists of
loops that never contract the distances of data values and keep the
order type of 
the data values constant. We then  define the notion of the type of a
run. It turns out that such a non-contracting loop exists if and only
if the automaton has a run realising a type among a certain
set. Finally, we provide a $\PSPACE$-algorithm that checks whether an
automaton realises a given type. Putting all these together yields our
main technical result.

\begin{theorem}\label{thm:EmptinessPSPACE}
  Emptiness of $\Tree^C$-constraint automata is in $\PSPACE$. 
\end{theorem}

\subsection{Emptiness and Stretching Loops}

We first introduce some notation before defining our notion of
stretching loop and characterising emptiness in terms of stretching
loops. 

From now on a \define{word} is always an element of $\Q^*$,
$\bigsqcap$ ($\sqcap$) denotes the (binary) \define{greatest common prefix
operator}, and
we fix a finite tuple of words $C=(c_1,c_2,\dots, c_m)$ called
constants. \textbf{We assume that $C$ is closed under prefixes.} Note that closing $C$ under prefixes results only in polynomial growth.

\begin{definition}
  Let $s_1, \dots s_n$ be
  constant symbols and $\sigma = \Set{\preceq, \KBleq, s_1, s_2, \dots,s_n}$.
  Given a tuple $\bar w = (w_1, w_2, \dots, w_n)$ of words,
  the \define{maximal common ancestor tree} of $\bar w$ is the
  $\sigma$-structure
  \begin{equation*}
  \MCAT(\bar w)= (M,{\preceq}\restriction{M^2},
  {\KBleq}\restriction{M^2}, w_1, w_2, \dots, w_n) ,
  \end{equation*}
  where $w_i$ is the interpretation of constant symbol $s_i$ and 
  \begin{equation*}
  M =
  \Set{\varepsilon} \cup \Set{ \bigsqcap_{i\in I} w_i | \emptyset
    \neq I\subseteq 
    \set{1, 2, \dots, n}}.    
  \end{equation*}
  The \define{(order) type $\typ(\bar w)$ of $\bar w$}
    is the $\sigma$-isomorphism type of $\MCAT(\bar w)$. 
   We set $\MCAT_C(\bar w) \coloneqq \MCAT(\bar w, C)$ and 
    $\typ_C(\bar w) \coloneqq \typ(\bar w, C)$.
\end{definition}
Labelling the words from $\bar w$ by constant symbols has the following
consequence: if $\typ_C(\bar w) = \typ_C(\bar v)$ for 
$\bar w= (w_1, w_2, \dots, w_n)$ then there is a unique 
isomorphism $h$ from $\MCAT_C(\bar w)$ to $\MCAT_C(\bar v)$ 
which maps $c\mapsto c$ for every $c\in C$ and $w_i \to v_i$ for
$w_i$ the $i$-th element of $\bar w$ and $v_i$ the $i$-th element of
$\bar v$.

\begin{definition}
  For $n\in\N$ we define a relation $\stretchleq$
  on configurations from $Q\times (\Q^*)^n$
  by $(q,\bar w) \stretchleq (p,\bar v)$ if 
  $q=p$,
  $\typ_C(\bar w) = \typ_C(\bar v)$ and the 
  induced isomorphism $h:\MCAT_C(\bar w) \to
  \MCAT_C(\bar v)$ satisfies for all $d,e\in \MCAT_C(\bar w)$
  if $d\prec e$  then $\lvert h(e) \rvert - \lvert h(d) \rvert
  \geq \lvert e \rvert - \lvert d\rvert$.
\end{definition}
Intuitively, $(q,\bar{w})\stretchleq (q,\bar{v})$ holds if both data
tuples
have the same order type and  the
lengths of intervals in $\MCAT_C(\bar v)$ seen as a subtree of
$\Q^*$ are greater than the lengths of the corresponding intervals in 
$\MCAT_C(\bar w)$. In the following sections, we make extensive use of the following
properties of $\stretchleq$.
\begin{lemma}\label{lem:stretchleqIsWQO} \label{lem:StrongUpwardsCompat}
  \label{lem:CommonStretchLeqBoundExists}
  \begin{enumerate}
  \item \label{lem:stretchleqIsWQOPart1} $\stretchleq$ is a well-quasi order.
  \item \label{lem:stretchleqIsWQOPart2} The (inverse) transition
    relation 
  $\to$ ($\to^{-1}$) is strongly upwards compatible with respect to
  $\stretchleq$ in the sense of \cite{FinkelS01}, i.e., if $u \to v$ ($u \to^{-1} v$) and $u \stretchleq u'$, then there is a $v'$ such that $v \stretchleq v'$ and $u' \to v'$ ($u' \to^{-1} v'$).
  \item  \label{lem:stretchleqIsWQOPart3}
    Given two configurations $(q,\bar w)$ and $(q, \bar v)$ such
    that $\typ_C(\bar w) = \typ_C(\bar v)$ then there is a
    configuration $(q, \bar u)$ such that $(q,\bar w) \stretchleq
    (q, \bar u)$ and $(q, \bar v) \stretchleq (q, \bar u)$. 
  \end{enumerate}
\end{lemma}

\begin{definition}
  A \define{loop} is a finite run $r = (c_i)_{i\leq n}$
  with $c_0 = (q,\bar{w})$, $c_n = (q,\bar{v})$ and $\typ_C(\bar w) = \typ_C(\bar v)$.  
  We say that a loop $r = (c_i)_{i\leq n}$ is 
  \define{stretching} if $c_0 \stretchleq c_n$.
\end{definition}

\begin{lemma}\label{lem:characteriseAcceptingRuns}
  Let $\Aut{A}$ be a constraint automaton. 
  $\Aut{A}$ has an accepting run if and only if there are partial runs
  $r_1$, $r_2$ where $r_1$ starts in an initial configuration and ends
  in some configuration $c$ whose state is a final state, and where
  $r_2$ is a stretching loop starting in $c$.
\end{lemma}
\begin{proof}
  $(\Rightarrow)$. Let $r=(c_i)_{i\in\N}$ be an accepting run.
  Since $r$ contains infinitely many configurations with a final state
  and $\stretchleq$ is a wqo,  we can find
  numbers $n_1 < n_2$  such that $c_{n_1}
  \stretchleq c_{n_2}$ 
  whence $(c_n)_{n\leq n_1}$, $(c_n)_{n_1 \leq n \leq n_2}$ are the
  desired runs. 
  
  \noindent $(\Leftarrow)$. Assume $r_1$ is a run from some initial
  configuration to $c_1$ whose state is a final state $f\in F$ and
  $r_2$ is a stretching loop starting in $c_1$ and ending in $c_2$.
  Since $c_1\stretchleq c_2$,
  iterated use of strong upwards compatibility
  (Lemma~\ref{lem:StrongUpwardsCompat}) yields runs $r_i$
  from $c_{i-1}$ to $c_i$ such that $c_{i-1}\stretchleq c_i$ for all
  $i\geq 3$. 
  Clearly, the composition of $r_1, r_2, r_3, r_4, \dots$ is an accepting
  run. 
\qed
\end{proof}

\subsection{Stretching Loops and Types of Runs}

\begin{definition}
  Let $r=(c_i)_{0\leq i \leq n}$ be a finite run,  
  with $c_0 = (q,\bar{w})$ and
  $c_n = (p, \bar{v})$. Setting 
  $\pi = \typ_C(\bar{w},\bar{v})$,  we say $r$ has \define{type} $\typ(r) =  (q, \pi, p)$.
\end{definition}

\begin{definition}
  Let $\bar w, \bar v$ be $k$-tuples of words such that $\typ_C(\bar
  w) = \typ_C(\bar v)$ and
  let  $h$ be the induced  isomorphism from $\MCAT_C(\bar w)$
  to $\MCAT_C(\bar v)$. $(\bar w, \bar v)$ is called
  \define{contracting} if one of the following holds.
  \begin{enumerate}
  \item There is some $d \in \MCAT_C(\bar w)$ such that $h(d) \prec
    d$.
  \item There are $d,e \in \MCAT_C(\bar w)$ such that $d \prec e$,
    $h(e) = e$ and $d \prec h(d)$.
  \end{enumerate}
  We call a loop $r$ from $(q, \bar w)$ to $(q, \bar v)$
  contracting if $(\bar w, \bar v)$ is contracting.  Otherwise,
  we call it (and its type) noncontracting.
\end{definition}
\begin{remark}
  The type of a loop determines whether it is noncontracting. Let us explain
  the term `contracting'. Fix a loop from $(q, \bar w)$ to $(q,
  \bar v)$. The  isomorphism $h:\MCAT_C(\bar w) \to
  \MCAT_C(\bar v)$ relates for every pair $\bigsqcap_{k\in K} w_k
  \prec \bigsqcap_{l\in L} w_l$ the interval $(\bigsqcap_{k\in K} w_k
  , \bigsqcap_{l\in L} w_l)$ with the interval $(\bigsqcap_{k\in K}
  v_k , \bigsqcap_{l\in L} v_l)$.  By definition, for
  every contracting loop there is a pair $(K,L)$ such that (  setting
  $\bigsqcap_{k \in \emptyset} w_k=\varepsilon$)
  \begin{equation*}
    \lvert  \bigsqcap_{l\in L} w_l \rvert - 
    \lvert \bigsqcap_{k\in K} w_k  \rvert  >
    \lvert  \bigsqcap_{l\in L} v_l \rvert - 
    \lvert \bigsqcap_{k\in K} v_k  \rvert.
  \end{equation*}
\end{remark}

The technical core of this section shows that if an automaton admits a
noncontracting loop then it admits a stretching loop with the same
initial and final state. This allows to rephrase the conditions
from Lemma \ref{lem:characteriseAcceptingRuns} in terms of types. 
The proof of this claim requires some definitions and
preparatory lemmas. 

\begin{definition}
  Let $u$ be a word and $m\in \N$. We define 
  the \define{insertion of an $m$-gap at $u$} to be
  $\iota_u^m: \Q^* \to \Q^*$ given by
  $\iota_u^m(w) =
  \begin{cases}
    w &\text{if }u\not\preceq w, \\
    u0^mv &\text{if } w = uv.
  \end{cases}$

  Given a finite run $r$, the 
  sequence $\iota_u^m(r)$ obtained by applying $\iota_u^m$ to each
  data value of  $r$ is  the
  \define{run obtained by insertion of an $m$-gap at $u$ in $r$}.   
\end{definition}
For $r=(c_i)_{i\in I}$ and $r'=(d_i)_{i\in I}$ we write
$r\stretchleq r'$ if $c_i \stretchleq d_i$ for all $i\in I$. 
Note that the insertion of a gap preserves $\preceq, \KBleq$ and
$\sqcap$ in both directions. 

\begin{lemma}\label{lem:insertionIsStretchleqCompatible}
  Given a run $r$ and a word $u$ such that $u$ is not
  a prefix of any 
  constant. The sequence $\iota_u^m(r)$  is
  indeed a run $r'$ of the same 
  type and $r \stretchleq r'$. 
\end{lemma}

Let $w, v\in \Q^*$. We say $w$ is \define{incomparable left} of $v$ if
$w\KBleq v$ and $w \not\preceq v$. In the same situation we
call $v$ \define{incomparable right} of $w$.

\begin{lemma}\label{lem:RightaboveShiftsRight}
  Let $\bar w, \bar v$ be $k$-tuples with
  $\typ(\bar w) = \typ(\bar v)$.  If $w_i$ is incomparable left
  (right) of $v_i$ and $v_i\preceq w_j$, then $w_j$ is incomparable
  left (right) of $v_j$ and  incomparable right (left)
  of $w_i$.
  %
\end{lemma}
\begin{proof}
  By type equality, we have that $v_i$ is incomparable left of
  $v_j$, whence the same holds for its descendant $w_j$. From
  $w_i \KBleq v_i \preceq w_j$ follows $w_i \KBleq w_j$, and $w_i \not\preceq
  w_j$ as $w_i \not\preceq v_i$.
  %
  %
  \qed
\end{proof}

\begin{proposition}\label{prop:noncontracting-to-stretching}
  Let $r$ be a noncontracting loop. There is a stretching
  loop $r'$ such that $r\stretchleq r'$.
\end{proposition}
\begin{proof}
  Let $r$ from $(q, \bar w)$ to $(q, \bar v)$ be a noncontracting loop and
  $h:\MCAT_C(\bar w) \to \MCAT_C(\bar v)$ the induced
  isomorphism.  
  We iteratively define a sequence $r = r_0 \stretchleq r_1
  \stretchleq \dots \stretchleq r_n $ of runs until $r_n$ is stretching.


  We call a pair $(u_1, u_2)\in \MCAT_C(\bar w)^2$ \emph{problematic} (with
  respect to $r$) if $u_1 \preceq u_2$ and
  $\lvert u_2 \rvert - \lvert u_1 \rvert > \lvert h(u_2) \rvert -
  \lvert h(u_1) \rvert$.
  Recall that in this case $u_2$ and $h(u_2)$ are not prefix of any
  constant $c$ from $C$ because $h$ fixes all such elements. Let
  $P_r$ be the set of all problematic
  pairs. We
  split the set of all problematic pairs into three parts, which we
  handle separately (cf.~Figure \ref{fig:Step1} for an example). Let
  \allowdisplaybreaks
  \begin{align*}
    L_r &= \Set{ (u_1, u_2) \in P_r |
        u_2 \text{ incomparable left of } h(u_2)},  \\
    R_r &= \Set{ (u_1, u_2) \in P_r |
        u_2 \text{ incomparable right of } h(u_2)}, \text{ and} \\
    D_r &= \Set{ (u_1, u_2) \in P_r |
        u_2 \text{ comparable to } h(u_2)}. 
  \end{align*}
  \noindent\textbf{L-Step:}
  If $L_r$ is nonempty, choose the $\KBleq$-minimal $u_2$ such that
  there is $u_1$ with $(u_1, u_2)\in L_r$.  Now fix $u_1$ such that
  $(u_1,u_2)\in L_r$ and
  $d \coloneqq (\lvert u_2 \rvert - \lvert u_1 \rvert) - (\lvert
  h(u_2) \rvert - \lvert h(u_1) \rvert)$
  is maximal.  Let $\iota = \iota_{h(u_2)}^d$ be the insertion of a
  $d$ gap at $h(u_2)$ and   $r' = \iota(r)$.
  Denote by $\iota(\bar{w})$ ($\iota(\bar{v})$) the data
  values of the first (last, respectively) configuration of $r'$.  Let
  $h':\MCAT_C(\iota(\bar{w})) \to \MCAT_C(\iota(\bar{v}))$ be the corresponding
  isomorphism. 
  %
  \setlength{\textfloatsep}{0.5\textfloatsep}
  \begin{figure}[t]
    \centering
    \begin{tikzpicture}[
  every node/.style={
    draw,
    font=\footnotesize,
    inner sep=0pt,
    outer sep=0pt,
    minimum size=8.5pt,
    label distance=3pt,
  },
  semithick,
  transform shape
  ]

  \newcommand{\mydot}{\tikz{\draw[fill=black] (0,0) circle[radius=1.15pt];}}

  \tikzset{h-image/.style={fill=black!25}}

  \pgfkeys{/tikz/mygrow/.code={\pgfkeysalso{grow via three points={%
          one child at (0,-#1) and two children at (0,-#1) and (0.7,-#1)}}}}

  \matrix[draw=none,rectangle,inner sep=0pt,column sep=10mm] {
    \node[circle,label=180:$u_1$] (root-1) {}
    [mygrow=1]
    child {
      node[circle,label=180:$u_2$] {\mydot}
      [mygrow=0.5]
      child {
        node[diamond,label=180:$x_1$] {}
        child {
          node[circle,h-image,label=180:$h(u_1)$] {}
          child {
            node[circle,h-image,label=180:$h(u_2)$] (h-u2-1) {\mydot}
            child {
              node[diamond,label=180:$x_2$] {\mydot}
            }
            child {
              node[diamond,h-image,label={[label distance=1pt]10:$h(x_1)$}] {}
              child[mygrow=1] {
                node[regular polygon,regular polygon sides=4,label=180:$y_1$] {}
                [mygrow=0.5]
                child {
                  node[diamond,h-image,label=180:$h(x_2)$] {\mydot}
                  child {
                    node[regular polygon,regular polygon sides=4,label=180:$y_2$] {\mydot}
                  }
                }
              }
              child {
                child {
                  node[regular polygon,regular polygon sides=4,h-image,label=0:$h(y_1)$] {}
                  [mygrow=1]
                  child {
                    node[regular polygon,regular polygon sides=4,h-image,label=0:$h(y_2)$] {\mydot}
                  }
                }
              }
            }
          }
        }
      }
    };
    &[-5mm]
  \node[circle] (root-2) {}
  [mygrow=1]
  child {
    node[circle] {\mydot}
    [mygrow=0.5]
    child {
      node[diamond] {}
      child {
        node[circle,h-image] {}
        [mygrow=1]
        child {
          node[circle,h-image] (h-u2-2) {\mydot}
          [mygrow=0.5]
          child {
            node[diamond] {\mydot}
          }
          child {
            node[diamond,h-image] {}
            child[mygrow=1] {
              node[regular polygon,regular polygon sides=4] {}
              [mygrow=0.5]
              child {
                node[diamond,h-image] (h-x2-2) {\mydot}
                [mygrow=0.5]
                child {
                  node[regular polygon,regular polygon sides=4] {\mydot}
                }
              }
            }
            child {
              child {
                node[regular polygon,regular polygon sides=4,h-image] {}
                [mygrow=1]
                child {
                  node[regular polygon,regular polygon sides=4,h-image] {\mydot}
                }
              }
            }
          }
        }
      }
    }
  };
  &
  \node[circle] (root-3) {}
  [mygrow=1]
  child {
    node[circle] {\mydot}
    [mygrow=0.5]
    child {
      node[diamond] {}
      child {
        node[circle,h-image] {}
        [mygrow=1]
        child {
          node[circle,h-image] {\mydot}
          [mygrow=0.5]
          child {
            node[diamond] {\mydot}
          }
          child {
            node[diamond,h-image] {}
            child[mygrow=1] {
              node[regular polygon,regular polygon sides=4] {}
              [mygrow=1]
              child {
                node[diamond,h-image] (h-x2-3) {\mydot}
                [mygrow=0.5]
                child {
                  node[regular polygon,regular polygon sides=4] {\mydot}
                }
              }
            }
            child {
              child {
                node[regular polygon,regular polygon sides=4,h-image] {}
                [mygrow=1]
                child {
                  node[regular polygon,regular polygon sides=4,h-image] (h-y2-3) {\mydot}
                }
              }
            }
          }
        }
      }
    }
  };
  &
  \node[circle] (root-4) {}
  [mygrow=1]
  child {
    node[circle] {\mydot}
    [mygrow=0.5]
    child {
      node[diamond] {}
      child {
        node[circle,h-image] {}
        [mygrow=1]
        child {
          node[circle,h-image] {\mydot}
          [mygrow=0.5]
          child {
            node[diamond] {\mydot}
          }
          child {
            node[diamond,h-image] {}
            child[mygrow=1] {
              node[regular polygon,regular polygon sides=4] {}
              [mygrow=1]
              child {
                node[diamond,h-image] {\mydot}
                [mygrow=0.5]
                child {
                  node[regular polygon,regular polygon sides=4] {\mydot}
                }
              }
            }
            child {
              child {
                node[regular polygon,regular polygon sides=4,h-image] {}
                [mygrow=1.5]
                child {
                  node[regular polygon,regular polygon sides=4,h-image] (h-y2-4) {\mydot}
                }
              }
            }
          }
        }
      }
    }
  };
  \\
};

\begin{scope}[on background layer={gray,dashed,semithick}]
  \draw (h-u2-1.north) -- (h-u2-2.north);
  \draw (h-u2-1.north) -- (h-u2-1.north -| h-u2-2.center);

  \draw (h-x2-2.north) -- (h-x2-3.north);
  \draw (h-x2-2.north) -- (h-x2-2.north -| h-x2-3.center);

  \draw (h-y2-3.north) -- (h-y2-4.north);
  \draw (h-y2-3.north) -- (h-y2-3.north -| h-y2-4.center);
\end{scope}

\end{tikzpicture}

    \caption{
      Example for Proposition~\ref{prop:noncontracting-to-stretching}:
      In the first tree $(u_1,u_2)$ is problematic , insertion of a
      gap (D-Step) at
      $h(u_2)$ makes (the pair corresponding to) $(x_1,x_2)$ problematic;
      insertion of a gap (L-Step) at $h(x_2)$ makes $(y_1,y_2)$ problematic;
      insertion of a gap (L-Step) at $h(y_2)$ makes the tree stretching.}
    \label{fig:Step1}
  \end{figure}
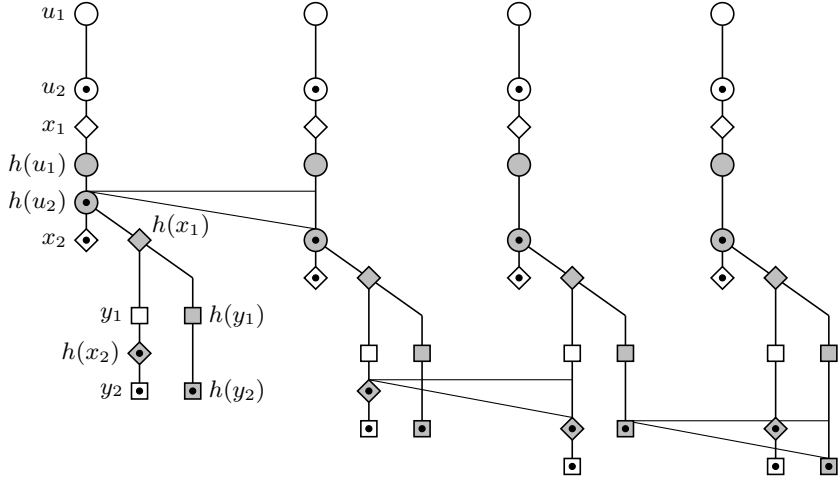
  By definition the set
  $L_{r'} = \Set{ (x_1, x_2) \in P_{r'} |
    x_2 \text{ incomparable left of } h'(x_2)}$  
  does not contain a pair $(u,\iota(u_2))$ for any $u
  \in\MCAT_C(\iota(\bar{w}))$.  Nevertheless, $r'$ may admit
  problematic pairs that are not problematic with respect to $r$.
  This can happen if there are $x_1,x_2\in \MCAT_C(\bar w)$ such that $x_1
  \prec h(u_2) \preceq x_2$ holds, but $h(x_1) \prec h(u_2) \preceq
  h(x_2)$ does not. Then, the distance
  between $\iota(x_1)$ and $\iota(x_2)$ is 
  greater than the distance between $x_1$ and
  $x_2$ (by $d$). On the other hand, either both or none of
  $h'(\iota(x_1))$ and $h'(\iota(x_2))$ are shifted by
  the insertion of the gap whence their distance
  is equal to the distance of  $h(x_1)$ and $h(x_2)$.
  
  In this case, possibly $(\iota(x_1), \iota(x_2))$ is problematic
  w.r.t.~$r'$ 
  while $(x_1, x_2)$ is not problematic w.r.t~$r$. Application of
  Lemma~\ref{lem:RightaboveShiftsRight} shows that then $x_2$ is
  incomparable left of $h(x_2)$ and $u_2$ is 
  incomparable left of $x_2$ whence 
  the same holds for $\iota(x_2), h'(\iota(x_2)) = \iota(h(x_2))$ and
  $\iota(u_2)$. 
  Thus, if $(\iota(x_1),\iota(x_2))$ is problematic, then
  $(\iota(x_1),\iota(x_2)\in L_{r'}$ and $\iota(u_2)$ is strictly 
  incomparable left of $\iota(x_2)$.

  Thus,
  iteration of this step  only creates  problematic pairs that
  are more and more to the right with respect to $\typ_C(\bar w_n) =
  \typ_C(\iota(\bar w))$.
  Since $\typ_C(\bar w_n)$ is finite, we eventually  
  do not introduce new
  problematic pairs and obtain a run $r_{i}$ such that  $L_{r_i} =
  \emptyset$  and $r \stretchleq r_i$ because 
  $r_i$ results from insertion of several gaps in $r$.

  \smallskip
  \noindent\textbf{R-Step}: 
  If $R_r\neq\emptyset$, proceed as in (L-Step)
  all  ``left'' and ``right''.

  \smallskip
  \noindent\textbf{D-Step}: 
  If $L_r = R_r = \emptyset$ and $r$ is not stretching, then
  $D_r \ne \emptyset$.  Choose $u_2$ $\KBleq$-minimal in $\MCAT(\bar w)$
  such that there is some $u_1$ with $(u_1,u_2)\in D_r$ and choose
  $u_1\prec u_2$ in $\MCAT_C(\bar w)$ such that
  $d \coloneqq (\lvert u_2 \rvert - \lvert u_1 \rvert) - (\lvert h(u_1) \rvert
  - \lvert h(u_2) \rvert)$
  is maximal. Since $r$ is not contracting we have
  $u_2 \preceq h(u_2)$ and $u_1\preceq h(u_1)$. Assume $u_2 = h(u_2)$,
  then $u_1 \prec h(u_1)$ as $(u_1,u_2) \in D$. This contradicts that
  $r$ is not contracting.  Thus $u_2 \prec h(u_2)$. Again, let
  $\iota = \iota_{h(u_2)}^d$ and $r' = \iota(r)$.

  Define $\iota(\bar{w}), \iota(\bar{v})$ and $h'$ as
  in the $L$-step.  Again there may be a pair $(x_1,x_2)$ which is
  not problematic with respect to $r$ while $(\iota(x_1), \iota(x_2))$
  is problematic with respect to $r'$.  If $R_{r'}$ or $L_{r'}$ are
  nonempty, we can deal with those problematic intervals using R- or
  L-steps.
  This finally leads to a run $r_{j}$ with
  $R_{r_j} = L_{r_j} = \emptyset$. Moreover, for every pair $(x_1,x_2)$
  such that this pair is not problematic with respect to $r$ but
  $(\iota(x_1), \iota(x_2))$ is problematic with respect to $r'$, we
  conclude that $x_2$ is strictly below $u_2$ whence 
  $\iota(x_2)$ is strictly below $\iota(u_2)$ 
  w.r.t.~$\preceq$.  Thus, the endpoints of problematic pairs move
  downwards 
  (in $\typ_C(\bar{w}, \bar{v}) = \typ_C(\bar w', \bar
  v')$) and  eventually all problematic pairs are removed. Once $r_{j}$
  is a loop without problematic pair, it is stretching.  \qed
\end{proof}

\begin{corollary}\label{cor:EmptinessTypeCharacerisation}
  The set of words accepted by an automaton $\Aut{A}$ is nonempty if and only if there are runs
  $r_1$ $r_2$ such that $r_2$ is a noncontracting loop starting in
  configuration $(f, \bar w)$ where $f$ is a final state 
  and $r_1$ is a run from an initial configuration to some
  configuration $(f, \bar v)$ such that $\typ_C(\bar w) = \typ_C(\bar
  v)$. 
\end{corollary}
\begin{proof}
  Due to  Lemma \ref{lem:characteriseAcceptingRuns}, only
  $(\Leftarrow)$ requires a proof.
  Assume that there are runs $r_1, r_2$ as stated above. 
  By Lemma \ref{lem:CommonStretchLeqBoundExists}, there is a run
  $r_2 \stretchleq r_2'$ such that $(f,\bar v) \stretchleq c_0$
  for $c_0$ the initial configuration of $r_2'$. 
  Note that $r_2'$ is also noncontracting whence 
  by Proposition \ref{prop:noncontracting-to-stretching} there is
  a stretching loop $r_2''$ such 
  that $r_2'\stretchleq r_2''$. Hence this loop starts in some
  configuration $c_1$ such that $(f,\bar v)\stretchleq
  c_1$. Applying Lemma \ref{lem:StrongUpwardsCompat} to $r_1$ and
  $c_2$ we obtain a run $r_1'$ from an initial configuration to
  $c_2$. Thus, $r_1'$ and $r_2''$ match the conditions of Lemma
  \ref{lem:characteriseAcceptingRuns} which completes the proof. 
  \qed
\end{proof}

\subsection{Emptiness and Computation of Types}

In order to turn this characterisation of emptiness in terms of types
into an effective algorithm for the emptiness problem the last missing
step is to compute whether a given type is realised by
some run of a given automaton.

For this purpose, we equip the set of all sets of types with a product
operation. Let $S,T$ be sets of types of runs; a type $(q,\pi,p)$ is
in $S \cdot T$ if there are $(q,\pi_1,r)\in S$, $(r,\pi_2,p)\in T$ and
tuples $\bar u,\bar v, \bar w$ such that $\typ_C(\bar u,\bar v) =
\pi_1$, $\typ_C(\bar v, \bar w) = \pi_2$ and $\typ_C(\bar u, \bar w) =
\pi$. Let $T_1$ denote the set of all types of runs of length $1$ (of
some fixed automaton $\Aut{A}$) and
$T_1^+ = \bigcup_{n\in\N} (T_1)^n$. By induction on the length, one
easily shows that every finite run $r$ of $\Aut{A}$ satisfies  $\typ(r) \in
(T_1)^+$.  Conversely, 
for every type $t\in (T_1)^+$ there is also a run of $\Aut{A}$ of type
$t$. This is due to the fact that gap-insertion preserves types
(Lemma~\ref{lem:insertionIsStretchleqCompatible}),
$\to$ is upwards compatible (Lemma~\ref{lem:StrongUpwardsCompat}) and
that trees of a given type $t_1$ with large gaps have, for all order
types $t, t_2$ with $t \in \set{t_1}\cdot \set{t_2}$,  an extension to 
a tree witnessing  this product. The necessary proofs are not very
difficult but tedious and lengthy. 

We conclude that a
type $t$ is in $(T_1)^+$ if and only if $t$ is the type of some run of
$\Aut{A}$. Moreover, types of runs can be
represented in polynomial space (in terms of the constants and the
dimension of a given automaton) and the product of types can be
computed in $\PSPACE$. Thus, we can determine whether an
automaton $\Aut{A}$ realises a  type $t$ by guessing
types in $T_1$ and computing an element of their product until it matches $t$. 
This proves the following proposition.

\begin{proposition}\label{prop:TypesComputable}
  There is a
  $\PSPACE$-algorithm that, given a $\Tree^C$-constraint automaton
  $\Aut{A}$ and a type 
  $t$, determines whether there is a run of $\Aut{A}$ of type $t$. 
\end{proposition}
\noindent
Together with Corollary \ref{cor:EmptinessTypeCharacerisation} we
obtain an algorithm proving Theorem~\ref{thm:EmptinessPSPACE}.

\begin{proof}[of Theorem \ref{thm:EmptinessPSPACE}]
  By Corollary \ref{cor:EmptinessTypeCharacerisation} it suffices that
  the algorithm guesses a type $(i,\pi, f)$ and a noncontracting type
  $(f, \pi', f)$ such that $i$ is an initial state, $f$ is a final
  state, and the order type of the last elements of $\pi$ coincides
  with the order type of the first elements of $\pi'$, and then checks
  whether these types are realised by actual runs using the previous
  proposition. 
\qed
\end{proof}

\subsection*{Acknowledgement}
We thank Claudia Carapelle for extremely helpful discussions and proof
reading. 

\bibliographystyle{splncs03}
\bibliography{local-data-automata}

\newpage
\appendix

\section{Proof of Proposition \ref{prop:Depth2ConstraitnsSuffice}}

First we recall the proposition. 
\begin{proposition}
  There is a polynomial time algorithm that computes, on input a
  $\cLTL$-formula $\phi$ an equivalent $\cLTL$-formula $\psi$ such
  that $\psi$ does not contain terms of the form $\neX^i x$ with
  $i\geq 2$. 
  \qed 
\end{proposition}

\begin{proof}
  First, we can replace any occurrence of $\neX^i x \mathrel* \neX^j
  y$ by $\neX^{\min(i,j)} (\neX^{i-\min(i,j)}) \mathrel* (\neX^{j-\min(i,j)}
  y)$.  Now assume that there is a subformula of the form $\neX^i x \mathrel*
  y$ (the case $x \mathrel* \neX^j y$ is symmetrical).  Introducing fresh
  variables $y_0,y_1, \dots, y_{i-1}$ we replace this formula by the
  formula $x \mathrel* y_i$  and add the conjunct 
  $\Glob(y_0 = y \land \bigwedge_{j=1}^{i} y_j = \neX y_{j-1})$
  which is polynomial in $i$. Obviously, this replacement yields  an
  equivalent formula. Iterating   this process for all constraints, we
  obtain the desired formula $\psi$.
\qed
\end{proof}

\section{Missing part of Theorem \ref{thm:SatAndMCInPSPACE}}
\label{sec:Orders}

Let $\Struc{O} = ( \set{11,22}^*12, \KBleq)$ where $\KBleq$ denotes the 
lexicographical order. 
\begin{lemma} \label{lem:QAsLexOrder}
  $\Struc{O}$ and $(\Q, <)$ are isomorphic. 
\end{lemma}
\begin{proof}
  $\Struc{O}$ is countable and does not have endpoints because
  $(11^n12)_{n\in\N}$ forms a strictly descending sequence such that
  any element of $\Struc{O}$ is minorised by some element of the
  chain. Analogously, $(22^n12)_{n\in\N}$ is a strictly increasing
  sequence majorising every element. 
  Thus, it is left to show that $\KBleq$ is a dense order. Let
  $w,v\in \Struc{O}$ with $w \KBleq v$. Writing $w=w_1w_2\dots w_k$
  with $w_i\in\set{11,12,22}$ and
  $v=v_1v_2\dots v_l$ with $v_i\in\set{11,12,22}$ 
  let $i$ be minimal such that $w_i\neq v_i$. If $v_i = 12$ then
  $w_i=11$ and $w_1w_2\dots w_i (22)^{\abs{w}}12$ is between $w$
  and $v$. 
  If $v_i = 22$ and $w_i = 11 $ or $w_i = 12$ then 
  $w \prec w_1w_2\dots
  w_{i-1}22(11)^{\abs{v}}12\prec v$. 
\qed
\end{proof}

\begin{definition}
  For $\sigma$ some  signature and $\sigma$-structures
  $\Struc{A}$ and $\Struc{B}$ we say a homomorphism $h:\Struc{A} \to
  \Struc{B}$ is a \define{$\sigma$-injection} if it is injective and preserves
  the relations, functions and constants under preimages.
\end{definition}

\begin{lemma}\label{lem:PreserverKBandPreceqToKAryTree}
  Let $h: (\Q, <) \to \Struc{O}$ be an isomorphism. 
  The  extension $g: \Q^* \to (\set{11,22}^*12)^*$, 
  given by $g(q_1q_2\dots q_n) = h(q_1)h(q_2)\dots h(q_n)$
  is an
  $\set{\preceq, \KBleq}$-injection of $\Tree = (Q^*, \preceq, \KBleq)$ into
  $\Tree[2]=(\set{1,2}^*, \preceq, \KBleq)$. 
\end{lemma}
\begin{proof}
  Note that $g$ is injective: if $w$ is in $\image(g)$, then the
  number of occurrences of $12$  where $1$ occurs at an odd position
  determines the length of every preimage $v$ such that $g(v) =
  w$. It is then a routine check to prove uniqueness of $v$. 
  
  We next show that $g$ preserves $\preceq$ (in both directions).
  It is obvious from the definition that $w \preceq v$ implies $g(w) \preceq
  g(v)$. Now assume that $g(w) \preceq g(v)$. Due to the same argument
  as in the injectivity proof, this implies that $w = w_1w_2\dots w_k$,
  $v = v_1v_2\dots v_l$, $k\leq l$ and $h(w_i) = h(v_i)$ for every
  $1\leq i \leq k$. Since $h$ is injective, it follows that $w_i =
  v_i$ for all $i\leq k$ which implies $w \preceq v$. 

  Finally, we have to prove preservation of $\KBleq$.  
  For rational numbers $q_1, q_2$ we have $q_1 < q_2$ iff  $h(q_1)
  \KBleq h(q_2)$.
  From this it easily follows that
  for words $w,w'\in Q^*$ $w\KBleq w'$ if and only if
  $w\preceq w'$ or $w = viw_1$ and $w' = vjw_2$ for some $v\in\Q^*$
  and some $i<j$
  if and only if
  $g(w) \preceq g(w')$ or $g(w) = g(v)h(i)g(w_1)$ and 
  $g(w') = g(v)h(j)g(w_2)$  with $h(i) \KBless h(j)$ 
  if and only if  $g(w) \KBleq g(w')$. 
  \qed
\end{proof}

\section{Missing Proofs Concerning $\stretchleq$}

In this section we prove Lemma \ref{lem:stretchleqIsWQO}.
Part \ref{lem:stretchleqIsWQOPart1} is proved in Lemma
\ref{lem:App:stretchleqIsWQO}, 
Part \ref{lem:stretchleqIsWQOPart2}  in Lemma \ref{prop:StrongUpwardsCompat} and
Part \ref{lem:stretchleqIsWQOPart3} in Lemma \ref{lem:stretchleqUpperBoundsExist}.

\subsection{Proof of Part 1}

\begin{lemma}\label{lem:App:stretchleqIsWQO}
  $\stretchleq$ is a well-quasi order.
\end{lemma}
\begin{proof}
  Obviously, $\stretchleq$ is a quasi order. 

  Any infinite sequence $(\bar w^i)_{i\in\N}$ of $n$-tuples of words
  induces an infinite sequence $(\bar w^i, C)_{i\in\N}$.  The latter
  has an infinite subsequence $(\bar w^i, C)_{i\in I}$ such that
  for all $i,j\in I$ $\typ_C(\bar w^{i}) = \typ_C(\bar w^{j})$.  This
  implies that $\MCAT_C(\bar w^i)$ and $\MCAT_C(\bar w^j)$ are
  isomorphic for all $i,j\in I$ via an  isomorphism
  $\phi_{i,j}$.

  For every $i\in I$ we define a map $f_i: \MCAT_C(\bar w^i)^2 \to \N$
  by $(u,v) \mapsto \lvert u\rvert - \lvert u\sqcap v \rvert$.  Fix an
  $i_0 \in I$ and an enumeration of the domain of $f_{i_0}$.  This
  induces an enumeration of the domain of $f_i$ for every $i\in I$ by
  letting $(u,v)\in\dom(f_i)$ be the $k$-th element if
  $(\phi_{i,i_0}(u), \phi_{i,i_0}(v))$ is the $k$-th element of
  $\dom(f_{i_0})$.
  
  By Dickson's Lemma we find tuples $\bar w^j$, $\bar w^k$ ($j<k$)
  such that for all $(u,v)\in\MCAT_C(\bar w^j)$ $f_k( \phi_{j,k}(u),
  \phi_{j,k}(v)) \geq f_j(u,v)$. From this we immediately conclude that
  $\bar w^j \stretchleq \bar w^k$. 
\qed
\end{proof}

\subsection{Proof of Part 2}

We prepare the proof of strong upwards compatability of the transition
relation by formally proving the following intuition: if $\MCAT_C(\bar
w')$ has larger gaps than $\MCAT_C(\bar w)$ (seen as subtrees of
$\Q^*$), every extension of $\MCAT_C(\bar w)$ to a bigger tree induces
a corresponding extension of $\MCAT_C(\bar w')$ to a bigger tree of
the same order type. 

\begin{definition}
  For $D, E, F$ sets with $D \subseteq E$, we say $h:E \to F$ extends
  $g:D\to F$ if $h\restriction{D} = g$.
\end{definition}

\begin{lemma}\label{lem:AuxiliaryUpwardsCompatibiltyOfTypes}
  Let $\sigma=\set{\preceq, \KBleq, \sqcap}$ and
  $\bar w, \bar w'\in \Q^*$ be tuples such that $\bar w
  \stretchleq \bar w'$.  The isomorphism $h: \MCAT_C(\bar
  w) \to \MCAT_C(\bar w')$ extends to a $\sigma$-injection 
  $f: (\Q^*, \preceq, \KBleq) \to
  (\Q^*, \preceq, \KBleq)$.
\end{lemma}
\begin{proof}
  In order to simplify the notation, we assume without loss of
  generality that $C\subseteq \bar w$.  We
  define a family of $\sigma$-injections $f_j: \Q^{\leq j} \to
  \Tree$ such that  $f_j$ extends $h\restriction{M_j}$ where
  $M_j = \Set{w \in \MCAT_C(\bar w) | \abs{w} \leq  \Q^{\leq j}}$.
  Let $f_0:\set{ \varepsilon} \to \set{ \varepsilon}$. 
  Assume that
  $f_j$ has been defined and satisfies that for all $\bar v \subseteq
  \bar w$ and all $u\in \Q^j$ 
  \begin{enumerate}
  \item $u \preceq \bigsqcap \bar v$ iff $f_j(u) \preceq h(\bigsqcap
    \bar v)$ and
  \item if $u\preceq \bigsqcap \bar v$ then $ \lvert \bigsqcap \bar v
    \rvert - \lvert u \rvert \leq \lvert f(\bigsqcap \bar v) \rvert -
    \lvert f_j(u) \rvert$.
  \end{enumerate}
  For each word $u\in \Q^{j}$, we define the values of $f_{j+1}$ on $u\Q$
  according to the following rule.  
  Let $\bar v_1, \bar v_2, \dots, \bar v_m \subseteq \bar w$ be those
  subsets such that for each $i$ there is some $q_i\in \Q$ with
  $\bigsqcap \bar v_i = u q_i$. We can assume that $q_1 \leq q_2 \leq \dots
  \leq q_m$.
  Note that the second condition on $f_j$ implies that
  $f_j(u)$ and $h(\bigsqcap \bar v_i)$ have distance at least $1$
  whence there is some $q_i'\in \Q$ such that $f_j(u)q_i'\preceq
  h(\bigsqcap \bar v_i)$. 
  We claim that for all $k,l \leq m$ we have $q_k \leq q_l$ if and
  only if $q_k' \leq q_l'$. 
  \begin{itemize}
  \item If $q_k = q_l$ then $u q_k \preceq \bigsqcap \bar v_k \sqcap
    \bigsqcap \bar v_l = \bigsqcap (\bar v_k \cup \bar v_l)$.
    Thus, there is some $i$ such that 
    $\bigsqcap \bar v_i = \bigsqcap (\bar v_k \cup \bar v_l)$ and 
    $q_i = q_k = q_l$. Then $f_j(u)q'_i \preceq h(\bigsqcap \bar v_i)
    \preceq h(\bigsqcap \bar v_k)$ and analogously for $\bar v_l$
    whence $q'_i = q'_k = q'_l$. 
  \item If $q_k < q_l$ then $\bigsqcap \bar v_k \sqcap \bigsqcap \bar
    v_l = u$. Thus, $u\in\MCAT_C(\bar w)$ and $f_j(u) = h(u) =
     h(\bigsqcap \bar v_k) \sqcap  h(\bigsqcap \bar v_l)$. 
    Moreover, $\bigsqcap \bar v_k \KBless \bigsqcap \bar v_l$ whence
    $h(\bigsqcap \bar v_k) \KBless h(\bigsqcap \bar v_l)$. 
    The only possibility to match both requirements is that 
    $q'_k < q'_l$. 
  \end{itemize}
  
  Fixing isomorphisms
  $g_i:\Set{ q\in Q | q_i <q < q_{i+1}} \to \Set{ q\in Q | q_i' <q <
    q_{i+1}'}$ (with $q_0 = q_0' = -\infty$ and $q_{m+1}= q'_{m+1} =
  \infty$), we define  for every $q\in Q$
  \begin{equation*}
    f_{j+1}(u q) =
    \begin{cases}
      h(\bigsqcap v_i) &\text{if }q = q_i, \\
      f_j(u)g_{i-1}(q) &\text{otherwise, where } q_i\in\set{q_1, \dots, q_m,
        q_{m+1}} \text{\ is minimal with } q < q_i.
    \end{cases}
  \end{equation*}
  Assuming that $f_j$ preserves $\preceq, \KBleq$, and $\sqcap$ in
  both directions, it is not difficult to prove the same result for
  $f_{j+1}$. Thus, the limit of $(f_j)_{j\in\N}$ is the desired
  $\sigma$-injection $f$. 
\qed
\end{proof}

\begin{proposition}
\label{prop:StrongUpwardsCompat}
  $\to$ and $\to^{-1}$ are strongly upwards compatible with respect to
  $\stretchleq$.
\end{proposition}
\begin{proof}
  Given $k$-tuples $\bar w, \bar v, \bar w'$ and states $q,p$ such
  that there is a transition $(q, \bar w) \to (p, \bar v)$ and such
  that $\bar w \stretchleq \bar w'$ we have to show that there is
  some $\bar v \stretchleq \bar v'$ and a transition $(q, \bar w')
  \to (p, \bar v')$. 
  
  Since $\bar w \stretchleq \bar w'$, 
  the isomorphism $h: \MCAT_C(\bar w) \to \MCAT_C(\bar w')$ extends 
  (by Lemma~\ref{lem:AuxiliaryUpwardsCompatibiltyOfTypes}) to
  a  $\set{\preceq, \KBleq, \sqcap}$-injection
  $\hat h: \Q^* \to \Q^*$.
  Setting $v'_i = \hat h(v_i)$ for each $v_i\in \bar v$ we obtain with
  $\bar v' = (v'_1,  \dots, v'_k)$ 
  that $(p, \bar v ) \stretchleq (p, \bar v')$ and $(q,\bar w') \to
  (p, \bar v')$ as desired.

  The argument for
  $\to^{-1}$ is completely analogous.
\qed
\end{proof}

\subsection{Proof of Part 3}

Recall from Lemma \ref{lem:insertionIsStretchleqCompatible} that
insertion of an $n$-gap at some $u$ which is not prefixed by a
constant from $C$ preserves the type and leads to a $\stretchleq$
larger tuple. Iterated use of this lemma proves  Part
\ref{lem:stretchleqIsWQOPart3} of Lemma \ref{lem:stretchleqIsWQO},
which we restate in the following lemma.

\begin{lemma} \label{lem:stretchleqUpperBoundsExist}
  Given two configurations $(q,\bar w)$ and $(q, \bar v)$ such
  that $\typ_C(\bar w) = \typ_C(\bar v)$ then there is a
  configuration $(q, \bar u)$ such that $(q,\bar w) \stretchleq
  (q, \bar u)$ and $(q, \bar v) \stretchleq (q, \bar u)$. 
\end{lemma}
\begin{proof}
  Let $d\in \N$ be  maximal such that there are $x_1,x_2\in
  \MCAT_C(\bar w)$ with $x_1\preceq x_2$ and $\abs{x_2} - \abs{x_1} =
  d$. Inductively, from the $\preceq$-maximal elements to
  $\varepsilon$ we insert a gap of size $d$ at each $y\in\MCAT_C(\bar
  v)$ if $y$ is not prefixed by a constant from $C$. All these
  iterated insertions result finally in a tuple $\bar u$ such that 
  $(q, \bar v) \stretchleq (q, \bar u)$ and for all 
  $z_1,z_2\in \MCAT_C(\bar u)$ such that $z_1\preceq z_2$ and $z_2$ is
  not prefix of any constant from $C$, then $\abs{z_2}-\abs{z_1}\geq
  d$. Thus, by definition of $d$ also $(q, \bar w) \stretchleq (q,
  \bar u)$ holds as desired. 
\qed
\end{proof}

\section{Computation of Types}

The goal of this section is to prove Proposition
\ref{prop:TypesComputable}, i.e., to provide an algorithm that checks
whether a given type is realised by one of the runs of a given
$\Tree^C$-automaton. For this purpose we first 
fix an $n$-dimensional  $\Tree^C$-constraint automaton $\Aut{A}$ with state set
$Q$. We equip the power set of all types with a product operation as follows.

\begin{definition}
  \begin{itemize}
  \item   Let $\RunTypes$ denote the \define{set of all types} $(q,\pi,p)$ where
    $q,p\in Q$ and $\pi= \typ_C(\bar w, \bar v)$ where $\bar w$
    and $\bar v$ are  $n$-tuples of words. 
  \item   We equip the power set $2^\RunTypes$ with a \define{product}
    $\cdot$ as 
    follows. For $t = (q_1,\pi_1,p_1) ,u = (q_2, \pi_2,p_2)$,
    $v = (q_3,\pi_3,p_3) \in \RunTypes$ let $t\in \set{u} \cdot \set{v}$
    if
    \begin{enumerate}
    \item $q_1 = q_2$, $p_1 = p_3$, $p_2 = q_3$, and
    \item there are $n$-tuples $\bar x, \bar y, \bar
      z$ such that $\typ_C(\bar x, \bar y) = \pi_2$, 
      $\typ_C(\bar y, \bar z) = \pi_3$ and 
      $\typ_C(\bar x, \bar z) = \pi_1$.
    \end{enumerate}
    Generally, for $A,B\subseteq \RunTypes$ such that at least one of
    them is not a singleton, we define
    $A\cdot B = \Set{t\cdot u | t\in A, u\in B}$. 
  \item The \define{set of types of
      one-step runs}  $T_1\subseteq \RunTypes$ is given by
    $t=(q,\pi,p)\in T_1$ if there is a transition $(q,\beta,p)$ of
    $\Aut{A}$ such that $\pi$ satisfies $\beta$.
  \item Let $T_1^1=T_1$, $T_1^{n+1} = T_1^n T_1$, and $T_1^+ =
    \bigcup_{n \geq 1} T_1^n$. 
  \end{itemize}
\end{definition}
\begin{remark}
  One easily checks that  $t\in T_1$ holds if and only if there is a
  run of length $1$  with type $t$. 
\end{remark}

The product operation resembles the composition of types. As a
consequence one can connect the runs of $\Aut{A}$ and $T_1^+$ as
follows. 

\begin{lemma}\label{lem:TypesAndRunsConnection}
  There is a run of $\Aut{A}$ of type $t$ if and only if $t \in
  T_1^+$. 
\end{lemma}

Before we provide a proof, we show how this lemma can be used to prove 
\ref{prop:TypesComputable} which we restate here:

\begin{proposition}
  There is a
  $\PSPACE$-algorithm that, given an $n$-dimensional
  $\Tree^C$-constraint automaton 
  $\Aut{A}$ and a type 
  $t$, determines whether there is a run of $\Aut{A}$ of type $t$. 
\end{proposition}
\begin{proof}
  Writing $m = \max(\abs{c}: c\in C)$ the algorithm uses polynomial
  space in terms of $m + n + \abs{\Aut{A}}$.\footnote{Assuming any
    reasonable notion of size of an automaton.}
  Given $n$-tuples $\bar w$ and $\bar v$, note that $\typ_C(\bar w, \bar v)$
  contains at most $4n$ elements that  are not constants. Thus, we can
  represent any type by $2$ states and $2n$ words of length at most 
  $m+ 4n$. 
  Moreover, It takes logarithmic space in $n$ and $\abs{A}$
  to check whether a given type satisfies a specific transition.
  Finally, it only needs $O( 2n (m+(4n\cdot 2n)))$ space to decide
  whether a given type $t$ is in the product of two types $t_1,t_2$
  (cf.~the upcoming Lemma \ref{lem:BigGapsAllowAllTypes}).
  
  Thus, an $\NPSPACE$ ( = $\PSPACE$) algorithm can guess a
  first type $t_1\in T_1$ and, having stored a type $t_i\in T_1^i$,
  it can guess another type $t\in T_1$ and a type $t_{i+1}$ and
  verify that $t_{i+1} \in \set{t_i} \cdot \set{t}$. This procedure is
  iterated until $t_{i}$ is the desired type and the algorithm reports
  that $t_i$ can be realised by some run. 
\qed
\end{proof}

\subsection{Proof of Lemma \ref{lem:TypesAndRunsConnection}}

We finally have to prove the connection between composition of runs
and products of their types. One direction is easily shown and
contained in the following lemma.

\begin{lemma}
  For $r=(c_i)_{1\leq i \leq n }$  a run (with $n\geq 2$), 
  $\typ(r)\in T_1^{n-1}$. 
\end{lemma}
\begin{proof}
  For $n = 2$ the claim follows by definition of $T_1^{2-1} = T_1$. 
  We proceed by induction.
  Write $c_i = (q_i, w^i_1, \dots, w^i_k)$.  Let
  $r'=(c_i)_{1\leq i \leq n-1}$ and $r_{n-1} = (c_i)_{n-1\leq i \leq n}$.
  By induction hypothesis  $\typ(r') = (q_1, \pi, q_{n-1}) \in T_1^{n-2}$  with 
  \begin{equation*}
    \pi = \typ_C(w^1_1, w^1_2, \dots, w^1_k, w^{n-1}_1,
    \dots, w^{n-1}_k),    
  \end{equation*}
  and $\typ(r_{n-1}) = (q_{n-1}, \pi_{n-1}, q_n) \in T_1$ with
  \begin{equation*}
    \pi_{n-1} =
    \typ_C(w^{n-1}_1, \dots, w^{n-1}_k, w^n_1, \dots, w^n_k).    
  \end{equation*}
  Thus, the tuples $w^1_1, \dots, w^1_k$, $w^{n-1}_1, \dots,
  w^{n-1}_k$, $w^n_1, \dots, w^n_k$ witness that
  \begin{equation*}
    (q_1, \pi', q_n) := \typ(r) \in \typ(r') \cdot \typ(r_{n-1})
    \subseteq T_1^{n-2} \cdot T_1 = T_1^{n-1}, 
  \end{equation*}
  which completes the proof.
\qed
\end{proof}

The other direction of Lemma \ref{lem:TypesAndRunsConnection} relies
on the following intuition.
\begin{enumerate}
\item By upwards-compatability and gap-insertion every type realised
  by some run, is realised by one with large gaps between all pairs of
  elements except the constants.
\item If two $n$-tuple $\bar w, \bar v$ have $2n$-gaps  between all
  pairs of elements from $\MCAT_C(\bar w, \bar v)$ except the
  constants,
  then for every type $t \in \typ_C(\bar w, \bar v) \cdot T_1$ there
  is a tuple $\bar u$ such that $\bar w, \bar v, \bar u$ witness this
  inclusion.
\item Thus, assuming that all types from $T_1^{n-1}$ are realised by
  runs, for all $t\in T_1^{n-1} \cdot T_1$ we can realise the
  appropriate type from $T_1^{n-1}$ with a run $r$ that has large gaps at
  its last configuration and find a witness for $t$ by realising the
  appropriate type  from $T_1$ using the values of the last
  configuration of $r$. 
\end{enumerate}
Proving these intuitions is rather tedious and we give the details in
the following. Recall that we assume that the set of constants $C$ is
closed under prefixes. Let us first make precise what a gap is. 

\begin{definition}
  We say that a tree $T\subseteq \Q^*$ has \define{$n$-gaps above C} if for
  all $d,e\in T$ with $d\prec e$ such that $e \not\preceq c$ for all
  $c\in C$ we have $\lvert e \rvert - \lvert d \rvert > n$.
\end{definition}

We can now give a precise version of the first claim. 
\begin{lemma}\label{lem:allRunsRealisedWithGaps}
  Given a finite run $r$ there is a run $r'$ from $c_1'=(q,\bar w')$ to
  $c_2 = (p', \bar v')$ of the same type such that 
  $\MCAT_C(\bar w', \bar v')$  has $2n$-gaps above $C$. 
\end{lemma}
\begin{proof}
  Let $r$ be a run from $(q,\bar w)$ to $(q,\bar v)$ For each $u \in
  \MCAT_C(\bar w, \bar v)$ (starting with $\preceq$-maximal ones) that
  is not a constant from $C$, we insert a gap of size $2n$ at $u$ in
  $r$ . Since gap insertion preserves types (Lemma
  \ref{lem:insertionIsStretchleqCompatible}), the resulting run $r'$
  from $(q, \bar w')$ to $(p, \bar v')$ is of the same type as $r$ and
  $\MCAT_C(\bar w', \bar v')$ has $2n$-gaps above $C$.  \qed
\end{proof}

For the second claim we need a technical lemma first and then prove
the second intuition to be correct.

\begin{lemma} \label{lem:Auxiliary:BigGapsAllowAllTypes} Let $\sigma =
  \set{\preceq, \KBleq, \sqcap}$, $n\in\N$.  Let $A\subseteq \Q^*$ be
  some finite set closed under maximal common prefixes such that
  $\varepsilon\in A$. Let $B\subseteq A$ and $h:A\to \Tree$ a
  $\sigma$-injection such that $h(A)$ has $n$-gaps above $h(B)$.
  Given $D\subseteq \Q^*$ such that
  \begin{enumerate}
  \item $\lvert D \setminus A \rvert \leq n$,
  \item $D\cup A$ is closed under maximal common prefixes, and
  \item there is no $d\in D$ and $b\in B$ such that $d\preceq b$,
  \end{enumerate}
  then $h$ extends to a $\sigma$-injection $h_D:A\cup D \to \Tree$.
\end{lemma}
\begin{proof}
  The base case $n = 0$ is trivial.  Assume that the lemma has been
  proven for some $n\in\N$.  If $\lvert D\setminus A \rvert = n+1$,
  let $d\in D\setminus A$ be $\KBleq$-minimal.  By induction
  hypothesis it suffices to extend $h$ to a $\sigma$-injection $h':
  A\cup \set{d} \to \Tree$ that has $n$-gaps above $h(B\cup\set{d})$.
  We first define the image of $d$ by a case distinction and prove
  that the resulting map $h'$ has the desired properties.  We
  distinguish two cases.
  \begin{enumerate}
  \item Assume that there is some $a\in A$ such that $d\preceq a$.
    Since $\varepsilon\in A$ we find a maximal $w\in A$ such that
    $w\prec d$. Moreover, $\bar a = \bigsqcap \Set{a\in A |d\preceq
      a}$ is well defined and satisfies $w\prec \bar a$.  Thus, $h(w)
    \prec h(\bar a)$ and there is a $q\in Q$ such that $h(w) q \preceq
    h(\bar a)$.  Let $h'$ be the extension of $h$ to $A\cup\{d\}$
    mapping $h'(d) = h(w)q$ and $h'(a) = h(a)$ for all $a\in A$.
  \item Otherwise, there is no $a\in A$ with $d\preceq a$. Let again
    $w\in A$ be maximal with $w\preceq d$ and let $q_d\in \Q$ such
    that $w q_d\preceq d$.  For later use we first establish that
    \begin{equation}
      \label{eq:wqdisAfree}
      \text{there is no $a\in A$ with $w q_d\preceq a$.}
    \end{equation}
    Assuming the contrary let $w q_d\preceq a$. Since $A\cup D$ is
    closed under maximal common prefixes, we conclude that $w q_d
    \preceq (a\sqcap d)\in A\cup D$. $(a\sqcap d)\in A$ contradicts
    the maximality of $w$. But due to $\KBleq$-minimality of $d$,
    $(a\sqcap d) \in D\setminus A$ is only possible if $d=a\sqcap d$
    which implies $d\preceq a$ which contradicts our assumption on
    $d$.

    We define a partition of $\Set{a\in A | w\prec a}$ by setting
    \begin{align*}
      A^- &= \Set{a\in A | w\prec a \text{ and }a\KBleq d}\text{ and}\\
      A^+ &= \Set{a\in A | w\prec a \text{ and } d\KBleq a}.
    \end{align*}
    If $A^-\neq \emptyset$ let $a^-$ be its $\KBleq$-maximal element.
    Since $h$ preserves $\prec$, there is some $q^-\in \Q$ such
    that $h(w)q^- \preceq h(a^-)$. If $A^- = \emptyset$ set
    $q^-=-\infty$.  Analogously, if $A^+\neq \emptyset$ let $a^+$ be
    its $\KBleq$-minimal element.  Since $h$ preserves $\prec$, 
    there is some $q^+\in \Q$ such that $h(w)q^+ \preceq h(a^+)$. If
    $A^- = \emptyset$ set $q^+ = \infty$.
    
    If $a^-$ and $a^+$ are both defined, we conclude with
    \eqref{eq:wqdisAfree} that there are $q_1 < q_d < q_2$ such that
    $wq_1\preceq a^-$ and $wq_2 \preceq a^+$. Since $h$ is a
    $\sigma$-injection, we directly conclude that $q^- < q^+$.

    Choose $q\in (q^-, q^+)$ arbitrarily and define the map $h': A\cup
    \set{d} \to \Tree$ by $h'(a) = h(a)$ for all $a\in A$ and $h'(d) =
    h(w)q$.

    We prepare the proof that $h'$ is a $\sigma$-injection by
    establishing that
    \begin{equation}
      \label{eq:wqdisAfreeHom}
      \text{for all $p\in (q^-,q^+)$ there is no $a\in A$ such that $h(w)p\preceq h(a)$.}
    \end{equation}
    Heading for a contradiction assume that there was such $a$ and
    note that $h(a^-) \KBless h(a) \KBless h(a^+)$ and $h(w)\prec
    h(a)$.  This would imply $a^-\KBless a \KBless a^+$ and $w\prec
    a$. But this clearly contradicts the definitions of $a^-$ and
    $a^+$ as maximal below $d$ (minimal above $d$, respectively).
  \end{enumerate}
  We claim that the resulting map $h'$ is a $\sigma$-injection.

  \noindent\textbf{Injectivity:}
  Heading for a contradiction, assume that there is an $a\in A$ with
  $h(a)=h(w)q$ then $h(w) \prec h(a)$ which implies $w\prec a$.  But
  then either $w\prec a\preceq d$ violates the choice of $w$ or
  $d\preceq a$. In the latter case the third condition on $D$ implies
  that there is no $b\in B$ with $a\preceq b$. But then $h(w)$ and
  $h(a)$ need to have an $(n+1)$-gap which is not the case. Thus, we
  have arrived at a contradiction and conclude that there is no $a\in
  A$ with $h(a) = h(w)q$ whence $h'$ is injective.

  \noindent\textbf{Preservation of $\preceq$:}
  We show that $h'$ preserves $\preceq$ in both directions. Choose
  some $a\in A$.
  \begin{itemize}
  \item If $a \preceq d$ then by choice of $w$ we have $a\preceq w$
    whence $h'(a) = h(a) \preceq h(w) \prec h'(d)$.
  \item If $h'(a) = h(a) \preceq h'(d) = h(w)q$, then $h(a) \preceq
    h(w)$ because $h'$ is injective. Thus, $a\preceq w \prec d$ as
    desired.
  \item If $d\preceq a$ we are in case one of the definition of
    $h'$. Thus, $\bar a \preceq a$ whence by definition $h'(d) \preceq
    h(\bar a) \preceq h(a) = h'(a)$.
  \item If $h'(d)=h(w)q \preceq h(a)$, we conclude with
    \eqref{eq:wqdisAfreeHom} that we are in case one of the definition
    of $h'$. Thus, $h(w)q\preceq h(\bar a) \preceq h(a)$ implies that
    $h(w)q \preceq h(a)\sqcap h(\bar a) = h(a\sqcap \bar a)$. Since
    $h$ is a $\sigma$-injection, it follows that $w\prec a\sqcap \bar
    a \preceq \bar a$. Since $d\preceq \bar a$, we obtain that
    $a\sqcap \bar a$ and $d$ are comparable. By maximality of $w$, we
    conclude $d\preceq (a \sqcap \bar a) \preceq a$.
  \end{itemize}

  \noindent\textbf{Preservation of $\KBleq$:}
  Due to the $\preceq$
  preservation, it suffices to prove preservation of $\KBless \cap
  \not\preceq$.  Again choose some $a\in A$.
  \begin{itemize}
  \item Assume that $a\KBleq d$ and $a\not\preceq d$. 
    If $a\KBleq w$ we immediately conclude that $h'(a) = h(a) \KBleq
    h(w) \KBleq h(w)q = h'(d)$. 
    Otherwise, one immediately concludes that $d\sqcap a = w$. 
    \begin{enumerate}
    \item If $h'$ has been defined in case one, we immediately
      conclude $a \sqcap \bar a = w$ and $a \KBless \bar a$ whence $
      h(a) \sqcap h(\bar a) = h(a \sqcap \bar a) = h(w)$ and $h(a)
      \KBless h(\bar a)$. Since $h(w) \prec h'(d) \preceq h(\bar a)$,
      it follows that that $h'(a) = h(a) \KBless h(w)$.
    \item Otherwise, $h'$ has been defined in the second case and we
      conclude that $a\in A^-$ whence $a  \KBleq a^-$. This implies
      that
      $h'(a) = h(a) \KBleq h(a^-) \KBleq h(w)q = h'(d)$. 
    \end{enumerate}
  \item Assume that $d\KBleq a$ and $d\not\preceq a$.
    First assume that $w\not\preceq a$. Then $d\sqcap a = w \sqcap a
    \prec w$ whence $w\KBleq a$. Since $h$ is a $\sigma$-injection, we
    obtain
    $h(w) \KBleq h(a)$, and $h(w) \sqcap h(a) = h(w\sqcap a) \prec
    h(w)$. Thus, $h(w) \preceq h'(d)$ directly implies 
    $h'(d) \KBleq h(a) = h(a')$. 
    Otherwise, we have $w\preceq a$. Since $d\KBleq a$ we conclude
    that  $w\prec a$. 
    \begin{enumerate}
    \item If $h'$ has been defined in case one, 
      $d\not\preceq a$, $w\prec a$ and maximality of $w$ imply that 
      $w = d \sqcap a = \bar a \sqcap a$.
      Since $\bar a$ and $d$ are on a common path, 
      we also have $\bar a \KBleq  a$. 
      Thus, $h(w) = h(\bar a \sqcap a) = h(\bar a) \sqcap h(a)$ and
      $h(\bar a) \KBleq h(a)$. Since $h'(d)$ and $h(\bar a)$ are on a
      common path, we obtain $h'(d) \KBleq h(a) = h'(a)$.
    \item Otherwise, $h'$ has been defined in case two. Then $w\prec
      a$ and $d \KBleq a$ imply $a^+ \KBleq a$. We conclude by choice
      of $q$ that
      $h'(d) = h(w)q \KBleq h'(a^+) \KBleq h(a)$. 
    \end{enumerate}
    Since $\KBleq$ is a total order, the
    backwards preservation of $\KBleq$ follows directly from the
    forward preservation: assume $h'(x) \KBleq h'(y)$, then forwards
    preservation and injectivity rules out the case $y \KBless x$, whence
    $x \KBleq y$ because $\KBleq$ is total. 
  \end{itemize}
  
  \noindent\textbf{Preservation of $\sqcap$:}  
  Finally, note that $h'$ preserves $\sqcap$ in both directions. Let
  $a\in A$.  If $a$ and $d$ are comparable, the claim follows from the
  preservation of $\preceq$. Otherwise, if $a$ and $d$ are
  incomparable (with respect to $\preceq$), then we conclude $a\sqcap
  d \in A$ whence $a\sqcap d = a \sqcap w$. But then also $h'(a)$ and
  $h'(d)$ are incomparable whence $h'(a) \sqcap h'(d) \preceq h'(w)$
  whence by definition of $h'(d)$ we have $h'(a) \sqcap h'(d) = h'(a)
  \sqcap h'(w) = h(a) \sqcap h(w) = h( a \sqcap w) = h'(a\sqcap w) =
  h'(a\sqcap d)$.
\qed
\end{proof}

\begin{lemma} \label{lem:BigGapsAllowAllTypes} Let $\bar w, \bar v$ be
  $n$-tuples and $t=(q,\pi,r), t_1 = (q,\pi_1,p), t_2 = (p, \pi_2, r)
  \in \RunTypes$ such that 
  $\typ_C(\bar w, \bar v) = \pi_1$,  and
  $\MCAT_C(\bar w, \bar v)$ has $(2n)$-gaps above
  $C$.  There  is an $n$-tuple $\bar u$ such that 
  $\typ_C(\bar v, \bar  u) = \pi_2$ 
  and $\typ_C(\bar w, \bar u) = \pi$.
\end{lemma}
\begin{proof}
  By definition of the product, there are $k$-tuples $\bar x, \bar y, \bar
  z$ such that $\typ_C(\bar x, \bar y) = \pi_1$, $\typ_C(\bar y, \bar
  z) = \pi_2$ and $\typ_C(\bar x, \bar z) = \pi$.  Fix the
  isomorphism
  $h: \MCAT_C(\bar x, \bar y) \to \MCAT_C(\bar w, \bar v)$.
  One shows by induction on $n$ that if $\MCAT_C(\bar x, \bar y)$ has $n_1\in\N$
  many leaves and $n_2\in \N$ many inner nodes then
  $\MCAT_C(\bar x, \bar y)$ has at most $n_1+n$ leaves and $n_2+n$
  inner nodes whence
  $\lvert \MCAT_C(\bar x, \bar y, \bar z) \setminus
  \MCAT_C(\bar x, \bar y) \rvert \leq 2n$.
  Thus,
  $h$ extends by 
  Lemma~\ref{lem:Auxiliary:BigGapsAllowAllTypes} 
  (setting $A = \MCAT_C(\bar x, \bar y)$, $B = C$, $D=\MCAT_C(\bar x,
  \bar y, \bar z) \setminus \MCAT_C(\bar x, \bar y)$, and seeing $h$ as
  an injection $A\to \Tree$) 
  to a
  $\Set{\preceq, \KBleq, \sqcap}$-injection
  $\hat h: \MCAT_C(\bar x, \bar y, \bar z) \to \Tree$ (which is the
  identity on all all constants from $C$) such that for
  $\bar u = \hat h(\bar z)$, $\typ_C(\bar w, \bar v, \bar u) =
  \typ_C(\bar x, \bar y, \bar z)$.  
  In particular, $\typ_C(\bar v, \bar u) = \pi_2$ and 
  $\typ_C(\bar w, \bar u) = \pi$ as desired. 
\qed
\end{proof}

Now we are prepared to prove the last direction of Lemma
\ref{lem:TypesAndRunsConnection} 

\begin{lemma}
  For every $t\in T_1^+$ there is a run $r$ of $\Aut{A}$ with
  $\typ(r)=t$. 
\end{lemma}
\begin{proof}
  As remarked before, for $t\in T_1^1 = T_1$ there is nothing to
  show. Let $r\in T_1^{n+1}$ and assume the claim is true for all
  $t\in T_1^n$. 
  Let $t\in t_1 \cdot t_2$ with $t_1\in T_1^n$ and $t_2\in T_1$ and
  let $r'$ be  a run of type  $t_1$. 
  Let $c_0=(q, \bar w)$ be the first and $c_1 = (p, \bar v)$ the last
  configuration of $r'$. 
  By Lemma~\ref{lem:allRunsRealisedWithGaps}, we can assume that 
  $\MCAT_C(\bar w, \bar v)$ has $2n$-gaps. 
  Thus, by Lemma~\ref{lem:BigGapsAllowAllTypes}, there is tuple $\bar
  u$  and a state $q'$ such that 
  $(p, \typ_C(\bar v, \bar u), q') = t_2$ and $(q, \typ_C(\bar w, \bar
  u), q') = t$. Thus, extending $r'$ by configuration $(q', \bar u)$
  results in the desired run $r$.
  \qed
\end{proof}

\end{document}